\documentclass[letterpaper]{article}
\usepackage{authblk} 

\usepackage{graphicx}          
\usepackage{amsmath,amssymb,amsfonts}
\usepackage{subcaption}
\usepackage{textcomp}

\newenvironment{proof}{%
  \vspace{-1em}\noindent\textit{Proof.}\ }%
  {\hfill$\square$\vspace{0.2em}}

\newcommand{\myvec}[1]{\boldsymbol{#1}}

\def\BibTeX{{\rm B\kern-.05em{\sc i\kern-.025em b}\kern-.08em
T\kern-.1667em\lower.7ex\hbox{E}\kern-.125emX}}

\newtheorem{lemma}{Lemma}
\newtheorem{proposition}{Proposition}
\newtheorem{remark}{Remark}

\title{Game-Theoretic Area Coverage Control with Cooperative-Adversarial Multi-Agent Systems\thanks{Preprint submitted to Automatica.}} 

\author[1]{Ruiming Zheng}
\author[1]{Mohammad Pirani}
\author[1]{Davide Spinello}

\affil[1]{\small Department of Mechanical engineering, University of Ottawa, Ottawa, ON\newline Email: \{rzhen014,mpirani,dspinello\}@uottawa.ca}

\begin{document}
\maketitle 

\begin{abstract}                       
We formulate a multi-agent area coverage control problem as a two-player zero-sum game between two agent groups with conflicting goals. Conventional coverage control allocates resources based on an environmental risk density field. In contrast, we generalize this metric by allowing a second group of adversarial agents to generate the spatial risk field. Coupled agent dynamics are linked through the area coverage metric, which functions as the game reward. This framework induces coupled gradient-descent-ascent controllers for the groups. Analysis of a low-dimensional case reveals a Hopf bifurcation dictated by the ratio of the groups' control gains. In the regime dominated by adversarial agents, the system is driven into a periodic chase–evasion cycle. In the regime dominated by ordinary agents, the system converges to a fixed configuration. Numerical simulations validate these theoretical insights. Finally, we characterize the Nash equilibrium conditions. Under this equilibrium, ordinary agents converge to a generalized centroidal Voronoi tessellation, whereas adversarial agents settle at their corresponding equilibrium centroids.
\end{abstract}

\section{Introduction}

Advancements in communication technology, computational power, and control theory have led to the adoption of multi-agent systems in various engineering applications. Using multiple agents, these systems can perform complex cooperative tasks that are difficult for a single agent to accomplish, such as harbor protection~\cite{simettiUseTeamUSVs2010,MiahSuruz2014NDoA}, coordinated surveillance~\cite{pimentaDecentralizedControllersPerimeter2013,milellaActiveSurveillanceDynamic2010}, and search and rescue missions~\cite{huMultiagentInformationFusion2013,pothalNavigationMultipleMobile2015}. However, with the popularization of multi-agent technology, research in this area has expanded beyond cooperation to include competitive scenarios, such as pursuit-evasion~\cite{BilginAhmetTunc2015Aatm}, perimeter defense~\cite{ShishikaDaigo2020ARoM}, and resilience to cyber-attacks~\cite{PiraniMohammad2023Gafa}.

In addressing the optimal deployment of the agents for the coverage problem, the environment is partitioned into regions that are ``best served'' by the agents to which they are assigned, resulting in a Voronoi tessellation of the workspace. The Voronoi diagram has a wide range of applications and a large number of variants; see~\cite{aurenhammerVoronoiDiagramsPartiallySupported2000,duAccelerationSchemesComputing2006,xiaoPedestrianFlowModel2016,wangFundamentalDiagramsPedestrian2019,xiaoInvestigationVoronoiDiagram2018,sudRealTimePathPlanning2008,degyvesProximityQueriesCrowd2013}. In essence, it is a partition of a space into sub-regions based on an ``effect'' from a given set of points. In the context of area coverage control for networked mobile sensors, this effect refers to the sensing performance functions of the agents. Cortes et al. introduced the notion of sensor coverage and utilized the idea behind Lloyd's algorithm~\cite{lloydLeastSquaresQuantization1982} to design the area coverage control for networked mobile sensors with time-invariant risk densities, which is a gradient descent search for the fixed points of the coverage metric. Importantly, when applied to the search for the optimal formations of mobile agents that maximize the coverage metric, Lloyd’s algorithm converges to centroidal Voronoi tessellations, in which agents occupy the centroids of cells forming a Voronoi partition of the workspace.

In practice, many phenomena cannot be modeled with static density fields, which motivates research on area coverage in time-varying environments. The resulting problem is non-autonomous due to the drift induced by the evolving environment. In such cases, Lloyd’s method is no longer applicable because the objective is no longer an optimal regulator problem—searching for fixed points of an autonomous coverage metric—but rather an optimal tracking problem, in which the optimal configuration evolves over time due to the time dependence of the risk field. In~\cite{leeControlledCoverageUsing2013} this is addressed by providing a solution to non-autonomous coverage control by proposing a feedback and feedforward control that is proven to track the trajectories of time-varying Voronoi centroids starting from a centroidal configuration. A generalization is proposed in~\cite{miahGeneralizedNonautonomousMetric2017}, with the development of a generalized Lloyd’s algorithm that unifies the transient and steady-state tracking of centroids varying over time, allowing agents to asymptotically achieve optimal configurations for a non-autonomous coverage metric. An application of this framework to evolving and diffusive environments is presented in~\cite{miahNonAutonomousCoverageControl2017}. Examples of using risk density to control the formation and action of agents can be found in~\cite{OurWorkCDC2024,luoDistributedCoordinationMultiagent2020,abdulghafoorDistributedCoverageControl2021}. In~\cite{soleymaniOptimalNonautonomousArea2023}, a model-free coverage control is proposed using adaptive dynamic programming, utilizing a reinforcement learning control that relies on neural networks to interpolate local data gathered and shared by each agent instead of a global model of the system. 

In this work, we reinterpret area coverage control as a two-player zero-sum game. Classical formulations assume that the risk density is predefined and is unaffected by the agents' behavior. However, in realistic scenarios, adversaries can observe and adapt to surveillance efforts by choosing strategies to avoid detection, such as pursuit-evasion games~\cite{garciaMultiplePursuerMultiple2021} and perimeter-defense games~\cite{leePerimeterdefenseGameAerial2020}. Motivated by this limitation, we consider adversaries that actively respond to the agents’ deployment and influence the coverage outcome.

The main contributions of this paper are summarized as follows:

\begin{itemize}
    \item We reformulate the area coverage control problem as a two-player zero-sum game between ordinary agents and adversarial agents, where the two populations pursue conflicting objectives with respect to a common coverage performance index. In contrast to existing formulations with agent-independent risk distributions, the proposed framework explicitly accounts for adaptive adversarial actions, thereby capturing key features of realistic surveillance and defense scenarios.
    

    \item We propose a pair of coupled gradient descent-ascent controllers, consisting of a Lloyd's gradient descent for the ordinary agents and a gradient ascent strategy for the adversarial agents. Through the analysis of a low dimensional case, we show that the resulting closed-loop system undergoes a Hopf bifurcation as the control gain ratio crosses a critical threshold, leading to a transition from periodic chase-evasion dynamics to stable fixed spatial configurations. Simulations of the full-dimensional system exhibit consistent behavior.
    
    \item We characterize the spatial configuration of the Nash equilibrium of the game, under which the ordinary agents form a centroidal Voronoi tessellation while the adversarial agents converge to their corresponding equilibrium points. In addition, we establish sufficient conditions for the existence of such an equilibrium.
\end{itemize}

\section{Problem Formulation}

Consider two sets of mobile agents of respective cardinality $M$ and $N$. The evolution of the two groups is dictated by conflicting strategies, as detailed below, and therefore we label them as ``ordinary'' and ``adversaries'', respectively. The two sets are networked with an underlying communication network structure, which in both cases is assumed to be fully connected and each node is assumed to have non-corrupted and updated knowledge of the state of the other nodes in the respective network. The reason for this strong assumption is to focus on the game theoretical formulation of the area coverage control problem, decoupling it from the technical treatment of more general cases of communication network topologies. Both sets are deployed in a bounded, convex two-dimensional spatial region $\Omega \in \mathbb{R}^2$. The collective state of the ordinary agents is denoted by $\myvec{P}(t):=[\myvec{p}_{1}(t),\dots,\myvec{p}_{M}(t)]^\top\in \mathbb{R}^{2M}$ and that of the adversaries by $\myvec{S}(t):=[\myvec{s}_{1}(t),\dots,\myvec{s}_{N}(t)]^\top\in \mathbb{R}^{2N}$, where $\myvec{p}_{i}(t)=[x^{p}_{i}(t), y^{p}_{i}(t)]^\top$, $\myvec{s}_{j}(t)=[x^{s}_{j}(t), y^{s}_{j}(t)]^\top$ are the positions of the $i$-th agent and $j$-th adversary, respectively.

The agents evolve according to the single integrator dynamics
\begin{subequations}\label{eq:integrators}
\begin{align}
\dot{\myvec{p}}_{i}(t) = \myvec{u}^{p}_{i}(t) \ \text{for} \ i\in\{1,\dots,M\}\\
\dot{\myvec{s}}_{j}(t) = \myvec{u}^{s}_{j}(t) \ \text{for} \ j\in\{1,\dots,N\}
\end{align}
\end{subequations}%
with velocity control inputs $\myvec{u}^{p}_{i}(t)=[u^{p}_{x i}(t), u^{p}_{\mathrm{yi}}(t)]^\top$ and $\myvec{u}^{s}_{j}(t)=[u^{s}_{x j}(t), u^{s}_{\mathrm{yj}}(t)]^\top$.

In area coverage control problems (see, e.g.,~\cite{cortesCoverageControlMobile2004,leeControlledCoverageUsing2013,miahGeneralizedNonautonomousMetric2017}), optimal spatial distributions of the agents correspond to centroidal Voronoi partitions with cells assigned to each agents. Formally, given the states $\myvec{P}$, the workspace $\Omega$ is partitioned into a set of corresponding Voronoi cells, such that $\Omega=\bigcup_{i=1}^{M}\mathcal{V}_i{(\myvec{P})}$. The Voronoi cell $\mathcal{V}_i$ assigned to the $i$-th agent $\myvec{p}_i$ is calculated by
\begin{align}
\mathcal{V}_i=\left\{\myvec{q} \in \Omega \mid f\left(r_i\right) \leq f\left(r_j\right), \forall{j \neq i}\right\} \label{eq:Voronoi}
\end{align}
where $\myvec{q}=(x,y)^\top$ is a point in the workspace $\Omega$, $r_i=\Vert \boldsymbol{q} - \boldsymbol{p}_i \Vert$ is the Euclidean distance between agent $\myvec{p}_i$ and $\myvec{q}$, and $f(\cdot)$ is the 
agents' performance function, which is chosen to be a smooth, monotonic function in the argument, so that the Voronoi partition results in non-overlapping polygonal cells~\cite{miahGeneralizedNonautonomousMetric2017}. In this work, we follow the setting in \cite{cortesCoverageControlMobile2004} and adopt the quadratic performance model
\begin{align}
f(r_i)={r_i}^2 \ \text{for} \ i\in\{1,\dots,M\} \label{eq:f}
\end{align}
so that a lower $f(\cdot)$ indicates better sensing performance. The definition in~\eqref{eq:Voronoi} induces order-1 Voronoi partitions, in which the agents are in one to one correspondence with the cells. Higher-order Voronoi diagrams have been used to improve redundancy and fault tolerance~\cite{liHigherOrderNonAutonomousOptimal2024}, while other interesting variants have been proposed for coverage control in cluttered environments~\cite{abdulghafoorDistributedCoverageControl2021}, where Voronoi cells are truncated at obstacles to generate collision-free trajectories.

The presence of the adversary agents in the workspace can be modeled in terms of a region of influence associated with each adversary,  which is coupled with the dynamics of the ordinary agents, in the sense that the two groups have different strategies that mutually affect the dynamics of the two networks. Following the setting in~\cite{miahGeneralizedNonautonomousMetric2017,MiahSuruz2014NDoA}, we adopt a Gaussian function to model the region of influence of each adversary $\myvec{s}_j$
\begin{align}
\phi_{j}(\myvec{s}_j;\myvec{q})=\frac{\exp \left(-\frac{1}{2}\left(\myvec{q}-\myvec{s}_{j}\right)^\top \myvec{\Sigma}^{-1}\left(\myvec{q}-\myvec{s}_{j}\right)\right)}{2 \pi\sqrt{ \operatorname{det} \myvec{\Sigma}}} \label{eq:phij}
\end{align}
The covariance matrix has the diagonal structure  $\myvec{\Sigma}=\sigma^2 \myvec{I}_2$. Increasing the covariance $\sigma$ broadens the risk density function~\eqref{eq:phij}, reduces its peak height, and yields a flatter profile. The collective risk density $\bar{\phi}$ is given by
\begin{align}
\bar{\phi}(\myvec{S};\myvec{q})=
\sum_j\phi_j(\myvec{s}_j;\myvec{q}) \ \text{for} \ j\in\{1,\dots,N\}
\label{eq:phibar}
\end{align}

In optimal area coverage control, the performance index is the coverage metric, which encodes the additive contribution of the sensing performance of the agents, weighted by a risk density function that quantifies in a spatially distributed sense the relative importance associated to each point in the workspace. A non-uniform risk distorts the distribution of agents, allocating more resources in areas with higher risk~\cite{cortesCoverageControlMobile2004,miahGeneralizedNonautonomousMetric2017,MiahSuruz2014NDoA}. In this work, we generalize and re-interpret this idea, by considering the risk density as induced by the evolution of the adversaries. However, as opposed to existing literature, the adversaries have a strategy that may be partially known to the ordinary agents, and therefore the dynamics of the two groups couples via the area coverage metric as follows
\begin{align}
\mathcal{H}(\myvec{P},\myvec{S},\mathcal{V})=\sum_{i=1}^{M}\int_{\myvec{\mathcal{V}_i}}f(r_{i})\bar{\phi}(\myvec{S}) \ d\Omega \label{eq:H}
\end{align}
where $f(r_i)$ is the sensing performance function of the $i$-th agent~\eqref{eq:f}, and $\bar{\phi}$ is the collective risk density function~\eqref{eq:phibar} determined by the states of the adversaries. This formulation is adapted from the locational optimization problems in~\cite{cortesCoverageControlMobile2004}, which introduced the notion of sensor coverage and formalized the optimal sensor placement problem in non-uniform environments. In this work, a lower value of coverage metric $\mathcal{H}$ corresponds to better coverage, since~\eqref{eq:f} is a monotonically increasing function of the Euclidean distance $r_i$ between the agent's state $\myvec{p}_i$ and a point of interest $\myvec{q}$.

Since the two sets pursue conflicting objectives with respect to a common performance index, the problem is formulated as a two-player zero-sum Nash game, posed as a maximin optimization.

Due to the fact that the performance index~\eqref{eq:H} encodes the dynamics of two networks with conflicting strategies, the area coverage function is a game payoff. In this sense, the goal of the ordinary agents is to find an optimal spatial configuration that minimizes $\mathcal{H}$, thereby maximizing coverage performance, while the adversaries aim to position themselves to maximize $\mathcal{H}$ and degrade the coverage within the domain $\Omega$. This interaction is modeled as a two-player zero-sum Nash game, formalized by the following maximin optimization statement
\begin{align}
\max_{\myvec{S}\in\Omega}\min_{\myvec{P}\in\Omega}\mathcal{H}(\myvec{P},\myvec{S},\mathcal{V})\ \text{subject to~\eqref{eq:integrators}} 
\label{eq:maxminH}
\end{align}


\section{Control Design and Stationary Points}
\subsection{Gradient-based Controls}

A common approach for solving the maximin problem~\eqref{eq:maxminH} is the Gradient Descent–Ascent (GDA) algorithm~\cite{pmlr-v119-jin20e}, in which the ordinary and adversarial agents make control decisions that respectively descend and ascend along the gradient of~\eqref{eq:H}, such that the ordinary agents minimize while the adversaries maximize the game payoff, thereby achieving their conflicting objectives as stated in~\eqref{eq:maxminH}.

The derivative of the coverage metric $\mathcal{H}$ with respect to $i$-th agent's state $\myvec{p}_i$ is given by~\cite{miahGeneralizedNonautonomousMetric2017}
\begin{align}
\frac{\partial \mathcal{H}}{\partial \myvec{p}_i}=\int_{\mathcal{V}_i}  -2 \frac{\partial f\left(r_i\right)}{\partial {r_i}^2}\left(\myvec{q}-\myvec{p_i}\right) \bar{\phi}(\myvec{S};\myvec{q}) \ d\Omega
\label{eq:gradH1}
\end{align}
with $\partial f\left(r_i\right)/\partial {r_i}^2=1$ from \eqref{eq:f}. The mass and centroid for the $i$-th Voronoi cell are given by $m_i=\int_{\mathcal{V}_i} \bar{\phi} \ d \Omega$ and $\myvec{c}^p_i=\int_{\mathcal{V}_i} \myvec{q} \bar{\phi} \ d\Omega/m_i$~\cite{cortesCoverageControlMobile2004}, thus the gradient \eqref{eq:gradH1} can be written as
\begin{align}
\frac{\partial \mathcal{H}}{\partial \myvec{p}_i}=-2m_i(\myvec{c}^p_i-\myvec{p}_i)
\label{eq:gradH}
\end{align}
The control law used by the ordinary agents to minimize the metric $\mathcal{H}$ is
\begin{align}
\myvec{u}^{p}_i=k^p(\myvec{c}^p_i-\myvec{p}_i)
\label{eq:up}
\end{align}
with control gain $k^p>0$. This control law is also referred to as Lloyd's algorithm~\cite{cortesCoverageControlMobile2004}, which is a gradient descent method that iteratively seeks a local minimum of the coverage metric $\mathcal{H}$. When applied to the problem of optimal agent formation for coverage optimization, Lloyd’s algorithm converges to a centroidal Voronoi tessellation (CVT), in which each agent occupies the centroid of its associated Voronoi cell. The asymptotic stability of the trajectories generated by~\eqref{eq:up} in static environment has been established in \cite{cortesCoverageControlMobile2004} using LaSalle's invariance principle.

In order to achieve the objective of maximizing the coverage metric $\mathcal{H}$, the adversaries use a gradient ascent control. The derivative of $\mathcal{H}$ with respect to the $j$-th adversary's state $\boldsymbol{s}_j$ yields
\begin{align}
\frac{\partial \mathcal{H}}{\partial \myvec{s}_j}=
\sum_{i=1}^{M}\int_{\mathcal{V}_i} f(r_i)\phi_j\myvec{\Sigma}^{-1}(\myvec{q}-\myvec{s}_j) d\Omega
\label{eq:gradHs2}
\end{align}
where $\myvec{\Sigma}$ is the covariance matrix in~\eqref{eq:phij}.

We define the weighted risk density $\hat{\phi}_{ij}=f(r_i)\phi_j(\myvec{s}_j,\myvec{q})$, which can be interpreted as the risk density $\phi_j$ introduced by the $j$-th adversary weighted by the performance function $f(r_i)$ of the $i$-th agent. In analogy to the mass and centroid definition introduced above for the Voronoi cells, we define the mass and the centroid associated to $\hat{\phi}_{ij}$ as $\hat{m}_{ij}=\int_{\mathcal{V}_i}\hat{\phi}_{ij}\ d \Omega $ and $\hat{\myvec{c}}_{ij}=\int_{\mathcal{V}_i}\myvec{q}\hat{\phi}_{ij}\ d\Omega/\hat{m}_{ij}$, respectively.

With these definitions, the gradient in~\eqref{eq:gradHs2} can be rewritten as
\begin{align}
\frac{\partial \mathcal{H}}{\partial \myvec{s}_j}&=
\sum_{i=1}^{M}\hat{m}_{ij}\myvec{\Sigma}^{-1}(\myvec{c}^{s}_{j}-\myvec{s}_j)
\label{eq:gradHs3}
\end{align}
where $\boldsymbol{\Sigma}^{-1}$ is positive definite, and $\hat{m}_{ij}>0$. The details of the derivation are given in the Appendix A, with
\begin{align}
    \myvec{c}^s_j=\frac{\sum_{i=1}^{M}\hat{m}_{ij}\hat{\myvec{c}}_{ij}}{\sum_{i=1}^{M}\hat{m}_{ij}}
    \label{eq:csj}
\end{align}

In order to maximize the metric $\mathcal{H}$, the adversaries evolve along the trajectories generated by the inputs
\begin{align}
\myvec{u}^{s}_j=k^s(\myvec{c}^{s}_j-\myvec{s}_j) \ \text{for} \ i\in\{1,\dots,N\}
\label{eq:us}
\end{align}
with control gain $k^s>0$.

\subsection{System's Equilibria}

The ordinary agents and the adversaries evolve according to control laws~\eqref{eq:up} and~\eqref{eq:us}, respectively, with the ordinary agents seeking to minimize the coverage metric $\mathcal{H}$ while the adversaries aim to maximize it. The control inputs of the two groups are inherently coupled: the ordinary agents move toward the centroidal Voronoi tessellation weighted by the risk density $\bar{\phi}$, which is itself influenced by the spatial distribution of the adversaries. Conversely, the adversaries adjust their positions to increase $\mathcal{H}$, thereby altering the risk distribution. This adversarial motion shifts the Voronoi centroids $\myvec{c}^p_i$, increasing the Euclidean distances between the agents’ positions $\myvec{p}_i$ and their corresponding Voronoi centroids. As a result, the agents must exert greater control effort, as dictated by~\eqref{eq:up}, to track the moving centroidal Voronoi configuration.

Here, we show that the dynamical system governed by the coupled control laws~\eqref{eq:up} and~\eqref{eq:us} converges to a fixed configuration under a certain ratio of control gains $k_p/k_s$.
\begin{proposition}
    Let $\gamma = k_p/k_s$ be the ratio between the gains of the controls in~\eqref{eq:up} and~\eqref{eq:us}. Define the critical gain ratio as
\begin{align}\label{eq:gammac}
\gamma_c=\frac{
\sum_{j=1}^{N}\sum_{i=1}^{M}\hat{m}_{ij}\left\|\myvec{c}^{s}_{j}-\myvec{s}_j\right\|^2_{\myvec{\Sigma}^{-1}}}{\sum_{i=1}^{M}m_i\|\myvec{c}^{p}_i-\myvec{p}_i\|^2
}
\end{align}
where $\|\cdot\|^2_{\myvec{\Sigma}^{-1}}$ is the squared norm weighted by $\myvec{\Sigma}^{-1}$. If $\gamma > \gamma_c$, then the ordinary agents~$\myvec{P}$ and adversaries~$\myvec{S}$ converge to the largest invariant set 
\begin{subequations}\label{eq:equilibrium}
    \begin{align}
    &\myvec{p}_i=\myvec{c}^{p}_i\ \text{for}\ i=1,\dots M \label{eq:CVT} \\
    &\myvec{s}_j=\myvec{c}^{s}_j\ \text{for}\ j=1,\dots N 
    \end{align}
\end{subequations}
\end{proposition}
\begin{proof}
Consider the candidate Lyapunov function $V = \mathcal{H}$. Since both the performance function~\eqref{eq:f} of the ordinary agents and the collective risk density~\eqref{eq:phibar} of the adversaries in $\mathcal{H}$ are positive definite, $V$ is positive definite.

The time derivative of the coverage metric $\mathcal{H}$ is given by
\begin{align}\label{eq:Hdots1}
\dot{H}=\sum_{i=1}^{M} \frac{\partial \mathcal{H}}{\partial \myvec{p}_i} \cdot \dot{\myvec{p}}_i+\sum_{j=1}^{N} \frac{\partial \mathcal{H}}{\partial \myvec{s}_j} \cdot \dot{\myvec{s}}_j
\end{align}
Substituting the feedback laws \eqref{eq:up} and \eqref{eq:us} into $\dot{\mathcal{H}}$ yields
\begin{align}
\dot{\mathcal{H}} =-k^p\sum_{i=1}^{M}m_i\|\myvec{c}^{p}_i-\myvec{p}_i\|^2
+
k^s\sum_{j=1}^{N}\sum_{i=1}^{M}\hat{m}_{ij}\left\|\myvec{c}^{s}_{j}-\myvec{s}_j\right\|^2_{\myvec{\Sigma}^{-1}}
\end{align}
When $\gamma \geq \gamma_c$, it follows that $\dot{\mathcal{H}} \leq 0$. Therefore, by LaSalle's invariance principle~\cite{LaSalle}, the system trajectories converge to the largest invariant set where $\dot{\mathcal{H}} = 0$, corresponding to the equilibrium conditions in~\eqref{eq:equilibrium}.
\end{proof}

Here, one can show that the critical ratio $\gamma_c$ is finite. Since the Voronoi cells in~\eqref{eq:Voronoi} are convex and the density functions $\phi$ and $\hat{\phi}$ are non-negative in each cell, the corresponding centroids $\myvec{c}^p_i$ and $\hat{\myvec{c}}_{ij}$ lie in the $i$-th Voronoi cell. Moreover, the centroid $\myvec{c}^s_j$ associated with $\myvec{s}_j$ lies in the bounded workspace $\Omega$, because it is a convex combination of $\{\hat{\myvec{c}}_{ij}\}_{i=1}^M$ as defined in~\eqref{eq:csj}. Hence all Euclidean distance terms in~\eqref{eq:gammac} are uniformly bounded by the diameter of $\Omega$, and therefore $\gamma_c$ is finite. In particular, on any feasible set where the denominator in~\eqref{eq:gammac} is bounded away from zero, there exists a finite worst-case value for the gain ratio $\gamma$ ensuring the convergence to a fixed configuration.

\begin{remark}
The control gain ratio $\gamma$ compares how quickly the regular agents can react to deviations from the desired configuration relative to how aggressively the adversary can disrupt the system, and thus naturally reflects the relative physical capabilities of the agents. In particular, large values of $\gamma$ are expected in scenarios where regular agents possess higher actuation authority, faster sensing and control update rates, or more efficient information exchange than the adversary. Conversely, small values of $\gamma$ correspond to adversarially dominated regimes, such as those involving delayed feedback, limited control authority of the regular agents, or highly agile adversarial dynamics.
\end{remark}

It is worth noting that the equilibrium described in~\eqref{eq:equilibrium} for the gradient descent–ascent dynamics does not necessarily correspond to a Nash equilibrium of the zero-sum game~\cite{pmlr-v119-jin20e,holdingConvergenceSaddlePoints2014}. In particular, the second-order condition on the minimizer side, namely $\nabla^2_{pp}\mathcal{H}\succeq 0$, may fail to hold at the centroidal Voronoi configuration~\eqref{eq:CVT}, since the positive definiteness of $\nabla^2_{pp}\mathcal{H}$ can be guaranteed only in specific cases depending on the choice of risk density $\bar{\phi}$~\cite{cortesCoverageControlMobile2004,duCentroidalVoronoiTessellations1999,davydovSparsityStructureOptimality2020}. 

In addition, for strictly convex-concave functions, gradient methods converge to a saddle point~\cite{BealeE.M.L.1959SiLa}, whereas such a guarantee is no longer available when this structure is absent. Since Lloyd's algorithm~\eqref{eq:up} yields only a local minimum of the coverage metric~\cite{cortesCoverageControlMobile2004}, even when the gradient descent-ascent dynamics converges to a fixed configuration, the limit should be interpreted, at most, as a local Nash equilibrium, provided that the corresponding local second-order conditions are satisfied. Moreover, when the objective function is only non-strictly convex-concave, oscillatory solutions of the closed-loop system may arise~\cite{holdingConvergenceSaddlePoints2014}.

In the following section, we analyze a low-dimensional system to establish the conditions under which a Nash equilibrium exists. In this reduced setting, we further show that the resulting Nash equilibrium is global within the low-dimensional system. Moreover, we show that when the ratio $\gamma=k^p/k^s$ crosses the critical value $\gamma_c$ in \eqref{eq:gammac}, the system undergoes a bifurcation, leading to a periodic trajectory.

\section{Nash Equilibria and Bifurcation Analysis}

In order to establish closed form conditions for the existence of Nash equilibria for this class of systems, we analyze the two-dimensional case with one ordinary agent and one adversary, and establish the existence of a bifurcation that separates two classes of solutions. In this low-dimensional simplified scenario the agent and the adversary move along an one-dimensional domain $\Omega=[a,b]$.

From the definition of Voronoi partition~\eqref{eq:Voronoi}, the Voronoi cell for a single agent trivially consists of the whole domain, namely $\mathcal{V} = \Omega$. The corresponding differential game is thus formulated as
\begin{align}
\max_{s\in[a,b]}\min_{p\in[a,b]}\mathcal{H}(p,s)=\max_{s\in[a,b]}\min_{p\in[a,b]}\int_a^b|x-p|^2\phi(x;s)\ dx \label{eq:H2d}
\end{align}
where $p$ and $s$ denote the positions of the agent and adversary, respectively, along the one-dimensional domain $\Omega$, $|x-p|^2$ is the sensing performance function of $p$, and $\phi(x;s)$ is the risk density introduced by the adversary, modeled as a one-dimensional Gaussian density function centered at $s$ with variance $\sigma^2$.

The gradient of $\mathcal{H}$ is given by
\begin{align}
\nabla\mathcal{H}=
    \begin{bmatrix}
        \frac{\partial \mathcal{H}}{\partial p} \\
        \frac{\partial \mathcal{H}}{\partial s}
    \end{bmatrix}
    =
    \begin{bmatrix}
        -2\int_a^b(x-p)\phi(x;s)dx\\
        \frac{1}{\sigma^2}\int_a^b(x-p)^2(x-s)\phi(x;s)dx
    \end{bmatrix}
\label{eq:gradH2d1}
\end{align}
In analogy to the higher-dimensional case, the Voronoi mass and centroid are calculated as $m(s) = \int_a^b \phi(x;s) dx$ and $c^p(s) = \int_a^b x\phi(x;s) dx/m(s)$, respectively. Define the weighted density $\hat\phi=(x-p)^2\phi$, then the mass and centroid associated with $s$ are calculated by $\hat{m}(p,s)=\int_a^b \hat{\phi}\, dx$, $c^s(p, s) = \int_a^b x \hat{\phi}\, dx/\hat{m}$. $\nabla\mathcal{H}$ is rewritten as
\begin{align}
    \nabla\mathcal{H}=
    \begin{bmatrix}
        -2m(s) \bigl(c^p(s) - p\bigr) \\
        \frac{\hat{m}(p,s)}{\sigma^2} \bigl(c^s(p, s) - s\bigr)
    \end{bmatrix}
\end{align}

The low dimensional system thus evolves along the trajectories generated by the gradient-based controllers
\begin{align}
    \begin{bmatrix}
    u^{p}\\
    u^{s}
    \end{bmatrix}
    =
    \begin{bmatrix}
    \dot{p}\\
    \dot{s}
    \end{bmatrix}
    =
    \begin{bmatrix}
    k^p(c^p(s)-p)\\
    k^s(c^s(p,s)-s)\\
    \end{bmatrix}
    \label{eq:2d}
\end{align}
with control gains $k_p>0$, $k_s>0$.

In the following, we establish the two families of solutions controlled by the ratio $\gamma = k_p/k_s$.

\subsection{Saddle Point and local Nash Equilibrium}

We now show that the fixed point of \eqref{eq:gradH2d1} is a saddle point and further prove that it constitutes a Nash equilibrium of the low dimensional game~\eqref{eq:H2d}.

\begin{lemma}
When $\gamma>\gamma_c$, the lower dimensional system~\eqref{eq:H2d} admits a fixed point at 
\begin{align}\label{eq:2d fixed point}
    (c^p, c^s)^\top = \left( \frac{a + b}{2}, \frac{a + b}{2} \right)^\top
\end{align}
\end{lemma}
\begin{proof}
Under the hypothesis that $\gamma>\gamma_c$, the agents converge to the largest invariant set~\eqref{eq:equilibrium}.

Substituting \eqref{eq:2d fixed point} into \eqref{eq:gradH2d1} yields
\begin{align}
\nabla\mathcal{H}
    =\begin{bmatrix}
        -2\int_a^b \ (x-\frac{a+b}{2})\phi(\frac{a+b}{2}) \ dx \\
        \frac{1}{\sigma^2}\int_a^b \ (x-\frac{a+b}{2})^3\phi(\frac{a+b}{2}) \ dx
    \end{bmatrix}
\label{eq:gradH2d2}
\end{align}
Since the integrand is odd with respect to the center of the interval $[a,b]$, , both integrals vanish. Therefore, $\nabla\mathcal{H}=0$ at~\eqref{eq:2d fixed point}, implying that the candidate point~\eqref{eq:2d fixed point} is a fixed point of~\eqref{eq:H2d}.
\end{proof}
\begin{proposition}
The fixed point~\eqref{eq:2d fixed point} is a saddle point of the coverage function $\mathcal{H}$.
\end{proposition}
\begin{proof}
We show that the fixed point~\eqref{eq:2d fixed point} is a saddle point by proving that the Hessian matrix evaluated at this point is indefinite, i.e., it has eigenvalues of opposite signs.

At the fixed point~\eqref{eq:2d fixed point}, the Hessian of $\mathcal{H}$ is given by
\begin{align}
    \nabla^2\mathcal{H}=
    \begin{bmatrix}
        \mathcal{H}_{pp} & \mathcal{H}_{ps} \\
        \mathcal{H}_{sp} & \mathcal{H}_{ss}
    \end{bmatrix}
    =
    \begin{bmatrix}
        -2m(\frac{\partial c^p}{\partial p}-1) &  -2m\frac{\partial c^p}{\partial s}\\
        \frac{\hat{m}}{\sigma^2}\frac{\partial c^s}{\partial p} & \frac{\hat{m}}{\sigma^2} \bigl(\frac{\partial c^s}{\partial s}-1\bigr)
    \end{bmatrix}
\label{eq:HessianH}
\end{align}
where
\begin{subequations}\label{eq:gradients1}
\begin{align}
    \frac{\partial c^p}{\partial p } &=0 \\
    \frac{\partial c^s}{\partial s} &= \frac{1}{\hat{m}^2\sigma^2}\left(\int_a^b x^2\hat{\phi}\ dx \int_a^b \hat{\phi}\ dx-\left( \int_a^b x\hat{\phi}\ dx\right)^2\right) \label{eq:ds/ds1}\\
    \frac{\partial c^p}{\partial s} &= \frac{1}{m^2\sigma^2}\left(\int_a^b x^2\phi\ dx \int_a^b \phi\ dx-\left(\int_a^b x\phi\ dx\right)^2\right) \label{eq:dc/ds1}\\
    \frac{\partial c^s}{\partial p} &= \frac{2}{\hat{m}^2}\left(\int_a^b (x-p)^3\phi\ dx \int_a^b (x-p)\phi\ dx - \hat{m}^2\right) \label{eq:ds/dp1}
\end{align}
\end{subequations}
Here, $\partial c^p/\partial p = 0$ due to the single ordinary agent setup. In~\eqref{eq:ds/dp1}, all odd-order terms vanish because the integrands are odd functions and the integration domain is symmetric about the fixed point~\eqref{eq:2d fixed point}, which implies $\partial c^s/\partial p = -2$.

We convert the Gaussian density $\phi$ in \eqref{eq:gradients1} to the standard normal form $\psi(z)=\exp(-z^2/2)/\sqrt{2\pi}$ by performing a location–scale change of variables $z=(x-p)/\sigma$. This substitution centers at  $p$ and scales by $\sigma$, yielding integrals over $z\in[-\epsilon,\epsilon]$ with $\epsilon=|a-b|/(2\sigma)$, which is the ratio between half the length of the workspace $[a,b]$ and the standard deviation $\sigma$ of the risk density $\phi(s)$. Define the $k$-th order truncated moment of the standardized Gaussian density function $\psi(z)$ as
\begin{align}
    \mu_k:=\int_{-\epsilon}^{\epsilon}\ z^k\psi(z)\ dz
    \label{eq:muk}
\end{align}
Then the terms in \eqref{eq:gradients1} are given by $m=\mu_0$, $\hat m=\sigma^2\mu_2$, $\partial c^p/\partial p=0$, $\partial c^p/\partial s=\mu_2/\mu_0$, $\partial c^s/\partial s=\mu_4/\mu_2$, and $\partial c^s/\partial p=-2$, and the Hessian at the equilibrium~\eqref{eq:HessianH} is written as
\begin{align}
    \nabla^2\mathcal{H}=
        \begin{bmatrix}
        \mathcal{H}_{pp} & \mathcal{H}_{ps} \\
        \mathcal{H}_{sp} & \mathcal{H}_{ss}
    \end{bmatrix}=
    \begin{bmatrix}
        2\mu_0 & -2\mu_2\\
        -2\mu_2 & \mu_4-\mu_2
    \end{bmatrix}
\label{eq:HessianH2}
\end{align}
The determinant is given by
\begin{align}
    \det \nabla^2 \mathcal{H} = 2\mu_0(\mu_4 - \mu_2) - 4{\mu_2}^2
    \label{eq:detH}
\end{align}
Since $z \in [-\epsilon,\epsilon]$, we have $z^4\leq z^2 \epsilon^2$ and 
\begin{subequations}    \label{eq:ieq}
\begin{align}
    \mu_4 &= \int_{-\epsilon}^{\epsilon}\ z^4\psi(z)\ dz \leq \epsilon^2\int_{-\epsilon}^{\epsilon}\ z^2\psi(z)\ dz = \epsilon^2\mu_2 \\
    \mu_2 &= \int_{-\epsilon}^{\epsilon}\ z^2\psi(z)\ dz \leq \int_{-\epsilon}^{\epsilon}\ \epsilon^2\psi(z)\ dz = \epsilon^2\mu_0
\end{align}
\end{subequations}
Substitute \eqref{eq:ieq} into \eqref{eq:detH} to obtain
\begin{align}
    \det \nabla^2\mathcal{H} \leq -\mu_0(1 + \epsilon^2) < 0
\end{align}
Since the Hessian matrix at~\eqref{eq:2d fixed point} is real and symmetric, a negative determinant implies that its eigenvalues are of opposite signs. Therefore, the fixed point~\eqref{eq:2d fixed point} is a saddle point of~\eqref{eq:H2d}.
\end{proof}

\begin{proposition}
The saddle point~\eqref{eq:2d fixed point} constitutes a local Nash equilibrium of the differential game~\eqref{eq:H2d} if 
\begin{align}
    \mu_4(\epsilon)<\mu_2(\epsilon)
\end{align}
where $\mu_k(\epsilon)$ is defined in~\eqref{eq:muk}. 
\end{proposition}

\begin{proof}
We adopt the sufficient condition of local Nash equilibrium by verifying the first- and second-order sufficient conditions~\cite{pmlr-v119-jin20e,ratliffCharacterizationComputationLocal2013}. In particular, consider a two-player zero-sum game $\mathcal{H}:\boldsymbol{P}\times \boldsymbol{S}\rightarrow\mathbb{R}$, where $\boldsymbol{P}$ is a minimizer and $\boldsymbol{S}$ is a maximizer. Any stationary point $\nabla \mathcal{H}(p^\ast,s^\ast) = 0$ satisfying the following sufficient condition is a local Nash equilibrium:
\begin{align}\label{eq:Nash condition}
    \mathcal{H}_{pp}(p^\ast,s^\ast) \succeq 0,
    \qquad
    \mathcal{H}_{ss}(p^\ast,s^\ast) \preceq 0.
\end{align}
That is, at $(p^\ast,s^\ast)$, $\mathcal{H}$ is locally minimized with respect to $p$ and locally maximized with respect to $s$.

At the saddle point~\eqref{eq:2d fixed point}, the first-order condition is satisfied as shown above. Moreover, from~\eqref{eq:HessianH2}, the second derivative with respect to the minimizer is \begin{align} \mathcal{H}_{pp}=2\mu_0(\epsilon)>0. \end{align} The second derivative with respect to the maximizer is
\begin{align}
\mathcal{H}_{ss} = \mu_4(\epsilon)-\mu_2(\epsilon) = \int_{-\epsilon}^{\epsilon}z^2(z^2-1)\psi(z)\,dz
\end{align}
Therefore, whenever $\mu_4(\epsilon)<\mu_2(\epsilon)$, we have $\mathcal{H}_{ss}<0$. Hence, the payoff function is locally convex with respect to the agent state $p$ and locally concave with respect to the adversary state $s$ in a neighborhood of $(c^p,c^s)$. The sufficient conditions for a local Nash equilibrium are thus satisfied.
\end{proof}

\begin{remark}
The condition $\mu_4(\epsilon)<\mu_2(\epsilon)$ depends on the dimensionless parameter $\epsilon=(b-a)/(2\sigma)$, which compares the half-length of the workspace with the standard deviation of the Gaussian risk density.
Since
\begin{align}
    \frac{\partial\mathcal{H}_{ss}(\epsilon)}{\partial\epsilon}=2\epsilon^2(\epsilon^2-1)\psi(\epsilon)
\end{align}
the $\mathcal{H}_{ss}(\epsilon)$ is strictly decreasing on $(0,1)$ and strictly increasing on $(1,\infty)$. Moreover, $\mathcal{H}_{ss}(0)=0$ and $\mathcal{H}_{ss}(\epsilon)\rightarrow 2$ as $\epsilon\rightarrow\infty$. Therefore, $\mathcal{H}_{ss}=0$ has a unique positive root $\epsilon_c>1$, and the condition $\mu_4(\epsilon)<\mu_2(\epsilon)$ is equivalent to $0<\epsilon<\epsilon_c$. For the case considered here, numerical evaluation gives $\epsilon_c\approx 1.37$. Thus, this condition corresponds to a sufficiently large Gaussian spread relative to the half-length of the workspace.
\end{remark}

\begin{remark}
The above proof verifies the sufficient conditions for a local Nash equilibrium. In the present low-dimensional setting, the centroidal configuration can be calculated explicitly, and the corresponding centroid is located at the center of the workspace~\eqref{eq:2d fixed point}. This yields a unique equilibrium configuration in the reduced game~\eqref{eq:H2d}. Hence, the local Nash equilibrium characterized above is, in fact, global.

By contrast, this conclusion does not extend directly to the original higher-dimensional system. Since Lloyd's algorithm yields only a local minimum of the coverage metric, the corresponding Nash equilibrium, when it exists, can in general be interpreted at most as a local one. Moreover, because the Hessian block $H_{pp}$ is not guaranteed to be positive semi-definite in the original system, one must verify numerically at the fixed configuration that the Hessian blocks satisfy $H_{pp}\succeq 0$ and $H_{ss}\preceq 0$ in order to conclude that the limit configuration corresponds to a local Nash equilibrium.
\end{remark}

\subsection{Bifurcation Analysis of the Dynamical System}

Having established the Nash equilibrium in \eqref{eq:equilibrium}, we now analyze the bifurcation behavior of the gradient ascent–descent system \eqref{eq:2d}. For $\mu_4 > \mu_2$ (equivalently, $\epsilon > \epsilon_c$), we show that the system undergoes a Hopf bifurcation when the gain ratio $\gamma$ crosses a critical value $\gamma_c$, giving rise to a limit cycle. In contrast, for $\mu_4 < \mu_2$, the system remains stable for all admissible choices of the gain ratio $\gamma > 0$ and converges to a fixed configuration, which, in this particular low-dimensional setting, coincides with the Nash equilibrium established above.

\begin{proposition}\label{proposition:imaginary roots}
Consider the coupled dynamical system~\eqref{eq:2d}. Let $\gamma = k_p/k_s$ be the ratio between the agent's gain $k_p$ and the adversary's gain $k_s$.
For $\mu_4 > \mu_2$, a Hopf bifurcation occurs, and a limit cycle emerges with $\gamma < \mu_4/\mu_2-1$ .
\end{proposition}

\begin{proof}
Linearization of~\eqref{eq:2d} around the saddle point $(p,s)^\top=(c^p,c^s)^\top$ gives the Jacobian matrix
\begin{align}\label{eq:Jacobian1}
\myvec{J}=
    \begin{bmatrix}
    \frac{\partial \dot{p}}{\partial p} & \frac{\partial \dot{p}}{\partial s}\\
    \frac{\partial \dot{s}}{\partial p} & \frac{\partial \dot{s}}{\partial s}
    \end{bmatrix}
    =
    \begin{bmatrix}
    k^p(\frac{\partial c^p}{\partial p}-1) & 
    k^p\frac{\partial c^p}{\partial s}\\
    k^s\frac{\partial c^s}{\partial p} & 
    k^s(\frac{\partial c^s}{\partial s}-1)
    \end{bmatrix}
\end{align}
Substituting the expressions for the partial derivatives of $c^p$ and $c^s$ into \eqref{eq:Jacobian1} yields
\begin{align}\label{eq:Jacobian2}
\myvec{J}=
    \begin{bmatrix}
    -k^p & 
    k^p\frac{\mu_2}{\mu_0}\\
    -2k^s & 
    k^s(\frac{\mu_4}{\mu_2}-1)
    \end{bmatrix}
\end{align}

To establish the occurrence of a Hopf bifurcation, we analyze the eigenvalues of the Jacobian matrix~\eqref{eq:Jacobian2} and verify that a pair of complex conjugate eigenvalues cross the imaginary axis as $\gamma=\gamma_c$~\cite{alma991044077619705161}.

The eigenvalues $\lambda_{1,2}(\myvec{J})$ can be calculated from the standard expression
\begin{align}\label{eq:eigenvalues}
\lambda_{1,2}=
\frac{1}{2}\left(\operatorname{tr}\boldsymbol{J} \pm \sqrt{(\operatorname{tr}\boldsymbol{J})^2-4\det \boldsymbol{J}}\right)
\end{align}
we obtain the following expressions
\begin{subequations}\label{eq:trdet}
    \begin{align}
    &\operatorname{tr}\boldsymbol{J}=k_s\left(\frac{\mu_4}{\mu_2}-1 - \gamma\right)
    \label{eq:trJ}
    \\
    &\operatorname{det}\boldsymbol{J} =\gamma k_s^2\left(1+2\frac{\mu_2}{\mu_0}-\frac{\mu_4}{\mu_2}\right)
    \label{eq:detJ}
    \end{align}
\end{subequations}
where $\gamma = k_p/k_s > 0$ as the control gains are positive. When $\mu_4 > \mu_2$, we have $\mu_4/\mu_2 - 1 > 0$; thus, the sign of $\operatorname{tr}\boldsymbol{J}$ may change as $\gamma$ varies. Substituting $\gamma = \gamma_c = \mu_4/\mu_2 - 1$ into \eqref{eq:trJ}, we have $\operatorname{tr}\boldsymbol{J} = 0$.

We proceed to determine the sign of $\operatorname{det}\boldsymbol{J}$. By using the property $\psi'(z)=-z\psi(z)$ of the standard Gaussian density $\psi(z)$, integration by parts of $\mu_2$ and $\mu_4$ yields
\begin{subequations}\label{eq:mu}
\begin{align}
 &\mu_2=\int_{-\epsilon}^{\epsilon}\ z^2\psi(z)\ dz
 =-2\epsilon\psi(\epsilon)+\mu_0 \label{eq:mu2} \\
 &\mu_4
 =\int_{-\epsilon}^{\epsilon}\ z^4\psi(z)\ dz
 =-2\epsilon(\epsilon^2+3)\psi(\epsilon)+3\mu_0 \label{eq:mu4}
\end{align}
\end{subequations}
Substituting \eqref{eq:mu} into \eqref{eq:detJ}, we obtain
\begin{align}\label{eq:detJrewrite}
    \operatorname{det}\boldsymbol{J}=2 \gamma k_s^2\left (\frac{\epsilon\psi(\epsilon)\bigl(\mu_0(\epsilon^2-2) + 4\epsilon\psi(\epsilon)\bigr)}{\mu_0\bigl(\mu_0-2\epsilon\psi(\epsilon)\bigr)}\right)
\end{align}
For this ratio, we first determine the sign of the denominator. Since the standard normal density $\psi(z)$ is strictly increasing for $z < 0$ and strictly decreasing for $z > 0$, the expression %
\begin{align*}
    2\epsilon\psi(\epsilon)=\int_{-\epsilon}^{\epsilon}\ \psi(\epsilon)\ dz
\end{align*}
attains its minimum at the boundaries $z = \pm \epsilon$. Therefore, for all $z \in (-\epsilon, \epsilon)$, we have $\psi(z) > \psi(\epsilon)$, which implies
\begin{align*}
    \mu_0=\int_{-\epsilon}^{\epsilon}\ \psi(z)\ dz>\int_{-\epsilon}^{\epsilon}\ \psi(\epsilon)\ dz=2\epsilon\psi(\epsilon)
\end{align*}
Thus, we conclude that $\mu_0 - 2\epsilon\psi(\epsilon) > 0$, and consequently, the denominator of~\eqref{eq:detJrewrite} is positive for all $\epsilon > 0$.

To determine the sign of the numerator, we start from
\begin{align*}
    \frac{\partial \mu_0}{\partial\epsilon}=2\psi(\epsilon)
\end{align*}
Differentiating the numerator of~\eqref{eq:detJrewrite} with respect to $\epsilon$ yields
\begin{align*}
    \frac{\partial}{\partial \epsilon}\bigl(\mu_0(\epsilon^2-2)\psi(\epsilon) +4\epsilon\psi(\epsilon)\bigr)=\epsilon\bigl(\mu_0-2\epsilon\psi(\epsilon)\bigr)
\end{align*}
which, as established above, is strictly positive for $\epsilon > 0$, implying that the numerator is monotonically increasing with respect to $\epsilon$. Evaluating the numerator at $\epsilon = 0$ gives zero. Therefore, for all $\epsilon > 0$, the numerator is strictly positive, and since the denominator is also positive, it follows that $\operatorname{det}\boldsymbol{J} > 0$.

Substituting $\gamma = \mu_4/\mu_0 - 1$ into~\eqref{eq:eigenvalues}, we obtain
\begin{align*}
    \lambda_{1,2}=\frac{1}{2}\sqrt{-4\operatorname{det}\boldsymbol{J}}
\end{align*}
which indicates that the eigenvalues form a pair of purely imaginary complex conjugates. As the ratio $\gamma = \frac{k_p}{k_s}$ crosses the critical value $1 - \frac{\partial c^s}{\partial s}$, the eigenvalues cross the imaginary axis. This behavior characterizes a Hopf bifurcation~\cite{strogatzNonlinearDynamicsChaos2015,wigginsIntroductionAppliedNonlinear2003}, confirming that a limit cycle emerges when the gain ratio $\gamma$ drops below the critical value $\mu_4/\mu_2-1$.
\end{proof}

\section{Simulation Results}

This section presents simulation results of the original high-dimensional system~\eqref{eq:maxminH} to illustrate the theoretical predictions. Several key behaviors analytically identified in the two-dimensional analysis persist in the high-dimensional setting.

\begin{figure}[ht!]
    \centering
    \begin{subfigure}{0.45\textwidth}
        \includegraphics[width=\textwidth]{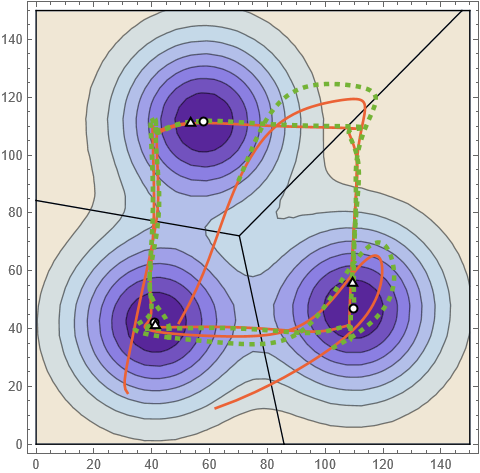}
        \caption{$k_p=1$, $k_s=1$, $\gamma=1$}
        \label{subfig:gamma=1 trajectories}
    \end{subfigure}
    \begin{subfigure}{0.45\textwidth}
        \includegraphics[width=\textwidth]{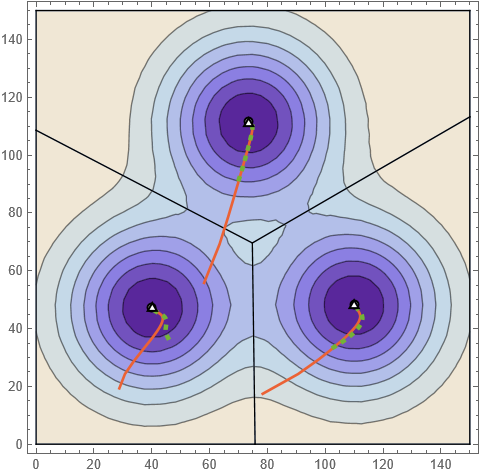}
        \caption{$k_p=3$, $k_s=1$, $\gamma=3$}
        \label{subfig:gamma=3 trajectories}
    \end{subfigure}
    \caption{Trajectories of adversarial (circles) and ordinary (triangles) agents under different control gain ratios. Adversarial motion induces the risk density, while ordinary agents generate the Voronoi partition. Green dashed and orange solid curves indicate adversarial and ordinary trajectories, respectively.}
    \label{fig:trajectories}
\end{figure}

\begin{figure}[ht!]
    \centering
    \includegraphics[width=0.85\textwidth]{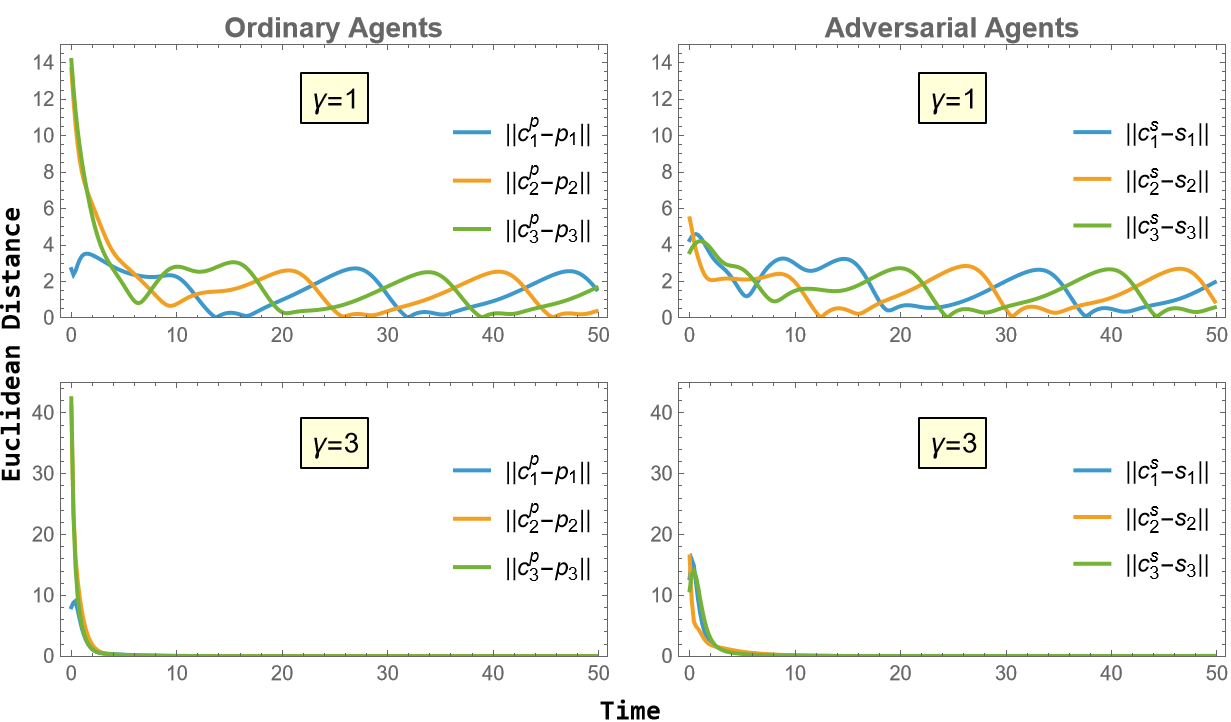}
    \caption{Evolution of the Euclidean distance between the ordinary and adversarial agents and their respective centroids.
    }
    \label{fig:distance}
\end{figure}

\begin{figure}[ht]
    \centering
    \includegraphics[width=0.6\textwidth]{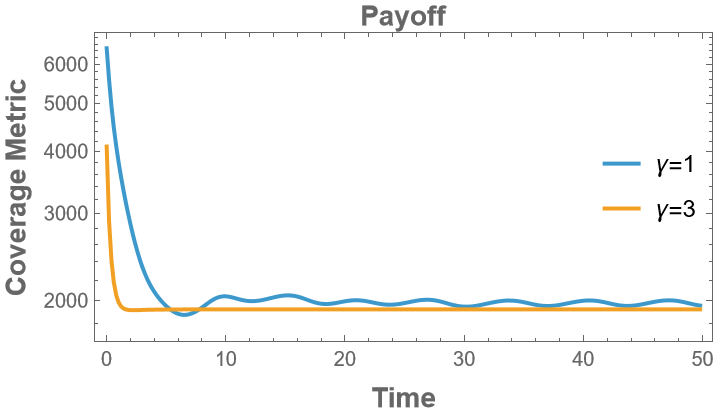}
    \caption{Evolution of game payoff under different control gain ratios.}
    \label{fig:H}
\end{figure}

\begin{figure}[ht!]
    \centering
        \includegraphics[width=0.8\linewidth]{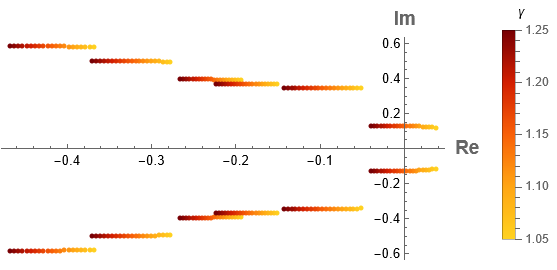}
        \label{subfig:eigenvalues}
    \caption{Eigenvalues' trajectory as gain ratio $\gamma$ varies with $3$ ordinary agents and $3$ adversary agents.}
    \label{fig:Eigenvalues}
\end{figure}

\begin{figure}[ht!]
    \centering
    \begin{subfigure}{0.45\textwidth}
        \includegraphics[width=\textwidth]{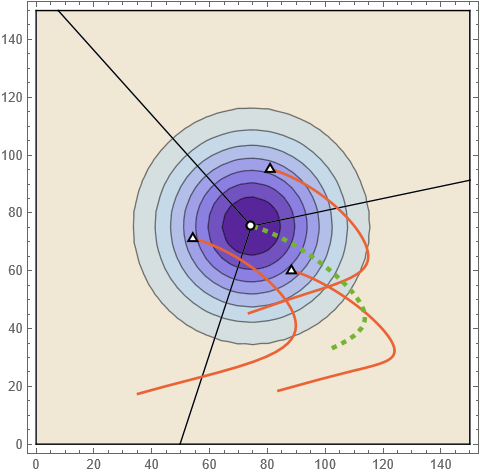}
        \caption{$\sigma=20$}
        \label{subfig:NonNE}
    \end{subfigure}
    \begin{subfigure}{0.45\textwidth}
        \includegraphics[width=\textwidth]{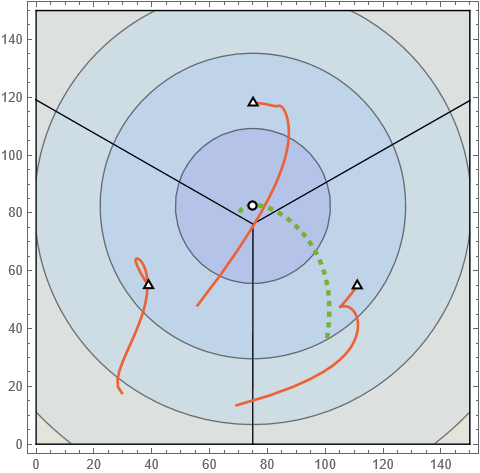}
        \caption{$\sigma=60$}
        \label{subfig:NE}
    \end{subfigure}
    \caption{Trajectories of three agents and one adversary. Both cases converge to fixed configurations, but only (b) satisfies the conditions of a Nash equilibrium.}
    \label{fig:3v1 NE}
\end{figure}

\subsection{Bifurcation of the Dynamical System}

This section focuses on the bifurcation arising in the dynamical system governed by the coupled controllers~\eqref{eq:up} and~\eqref{eq:us}. Simulations are conducted in a square planar domain $\Omega$ of side length $150$ (in a consistent system of units), involving three agents initially positioned at $(33, 17)$, $(45, 35)$, and $(54, 10)$, and an equal number of adversaries located at $(100, 32)$, $(45, 33)$, and $(69, 89)$. We first present two sets of simulation results for different bifurcation ratios $\gamma = k_p/k_s$ to illustrate the bifurcation of the dynamics.

Figure~\ref{fig:trajectories} compares the trajectories of ordinary and adversarial agents under different control gain ratios. In Fig.~\ref{subfig:gamma=1 trajectories}, the gains are set to $k_p = 1$ and $k_s = 1$, yielding a gain ratio of $\gamma = 1$. In this adversarially dominated regime, the coupled networks converge to a periodic chase-evasion motion characterized by a limit cycle. In Fig.~\ref{subfig:gamma=3 trajectories}, the control gain of the ordinary agents is increased to $k_p = 3$, resulting in $\gamma = 3$. In this case, the system converges to a fixed configuration.

Figure~\ref{fig:distance} illustrates the evolution of the Euclidean distances between the agents and their corresponding centroids. When the gain ratio is $\gamma = 1$, the distances exhibit sustained periodic oscillations. Increasing the gain ratio to $\gamma = 3$ leads to the convergence of the distances to zero, indicating stabilization of the system as the ratio between $k_p$ and $k_s$ increases. In this regime, the system converges to a fixed configuration, where the ordinary agents approach the centroidal Voronoi configuration $\myvec{c}^p$, and the adversarial agents converge to the corresponding equilibrium point $\myvec{c}^s$.

Figure~\ref{fig:H} shows the area coverage metric $\mathcal{H}$ for different bifurcation factors $\gamma$, plotted on a logarithmic scale. For $\gamma = 1$, $\mathcal{H}$ exhibits a fluctuating pattern due to the chase–evade interaction. As $\gamma$ increases to $3$, the coverage metric converges to a constant value of approximately $1900$.

Fig.~\ref{fig:Eigenvalues} shows the eigenvalue trajectory of the Jacobian matrix as the gain ratio $\gamma$ varies near the critical value. The fixed configuration is first computed at $\gamma=1.2$, and the Jacobian is then evaluated at this nominal configuration while treating $\gamma$ as a varying parameter. As $\gamma$ decreases, a pair of complex-conjugate eigenvalues crosses the imaginary axis, indicating the onset of oscillatory behavior consistent with a Hopf bifurcation.

\subsection{Nash Equilibrium}

It is important to emphasize that the stabilization of the dynamical system is not necessarily equivalent to a Nash equilibrium. To determine whether a converged fixed configuration constitutes a local Nash equilibrium, we examine the diagonal blocks of the Hessian matrix of the game payoff.
Specifically, the Hessian matrix of the payoff at the fixed point $(\myvec{c}^p_i,\myvec{c}^s_j)$, $\forall i\in\{1,\dots,M\},\ j\in\{1,\dots,N\}$, for the original system~\eqref{eq:H}, is given by
\begin{align}\label{eq:hessian}
    \nabla^2\mathcal{H}=
    \begin{bmatrix}
        \mathcal{H}_{pp} & \mathcal{H}_{ps} \\
        \mathcal{H}_{sp} & \mathcal{H}_{ss}
    \end{bmatrix}
    \in\mathbb{R}^{(2M+2N)\times(2M+2N)}
\end{align}
The details of each block are provided in the Appendix.

Figure~\ref{fig:3v1 NE} displays two sets of simulation results with three ordinary agents and one adversary. The agents and adversary are initialized with identical configurations and control gains, while the covariance $\sigma$ of the adversary’s risk density is varied.  

In Figure~\ref{subfig:NonNE}, the covariance is set to $\sigma=20$. The ordinary agents converge to $(54.1,71.2)$, $(80.8,95.2)$, and $(88.2,60.1)$, while the adversarial agent converges to $(74.3,75.5)$, for which the Hessian satisfies $\mathcal{H}_{pp}$ indefinite and $\mathcal{H}_{ss}$ negative definite. Hence, the fixed point does not constitute a Nash equilibrium.  

In contrast, when the covariance is increased to $\sigma=60$ in figure~\ref{subfig:NE}, thereby reducing the ratio between the size of the workspace and the spread of the Gaussian density, the ordinary agents converge to $(38.9,54.9)$, $(75.0,119.1)$, and $(111.1,54.9)$, while the adversarial agent converges to $(74.3,75.5)$.

At this fixed configuration, the corresponding Hessian matrix satisfies $\mathcal{H}_{pp} \succ 0$ and $\mathcal{H}_{ss} \prec 0$, confirming that the configuration constitutes a local Nash equilibrium of the two-player zero-sum game.

In summary, when the adversary's risk density is modeled with a large spread relative to the workspace---so that the density varies gently across the domain (in the two-dimensional case, $\epsilon<\epsilon_c$)---the coupled dynamical system is stable and converges to a fixed configuration, which corresponds to a Nash equilibrium of the differential game. In contrast, in the complementary regime (in the two-dimensional case, $\epsilon>\epsilon_c$), a Hopf bifurcation arises as the gain ratio $\gamma$ varies. In particular, large values of $\gamma$ are expected when the regular agents have greater actuation authority, faster sensing and control update rates, or more efficient information exchange than the adversary. Conversely, small values of $\gamma$ correspond to adversary-dominated regimes, e.g., due to delayed feedback, limited control authority of the ordinary agents, or highly agile adversarial dynamics.

\section{Conclusion}
This paper studied area coverage control in the presence of adaptive adversarial agents from a game-theoretic and dynamical systems perspective. By formulating the problem as a two-player zero-sum game and analyzing the resulting gradient-descent-ascent dynamics, we showed that the collective behavior of the system is governed by the control gain ratio between ordinary and adversarial agents. In particular, a Hopf bifurcation gives rise to a transition from periodic chase-evasion trajectories under adversarial dominance to stable fixed spatial configurations when ordinary agents possess sufficient control authority. We characterized the structure and existence conditions of equilibrium configurations, under which ordinary agents form centroidal Voronoi tessellations and adversarial agents converge to their corresponding equilibrium points. These theoretical predictions were further validated through numerical simulations. Future work includes extending this framework to settings with incomplete or uncertain information, in which ordinary and adversarial agents may operate under partial observability and limited knowledge of the environment or opposing strategies.

\appendix
\section*{Appendix}

\section{Gradient of the Game Payoff $\mathcal{H}$ with respect to Adversary State $\myvec{s}_j$}
\label{appendixA}

The payoff $\mathcal{H}$ depends on the adversary state $\myvec{s}_j$ through the risk density $\bar{\phi}$. Since
\begin{align}
    \frac{\partial \phi_j}{\partial \myvec{s}_j}
    =
    \phi_j \boldsymbol{\Sigma}^{-1}(\myvec{q}-\myvec{s}_j),
\end{align}
the gradient of $\mathcal{H}$ with respect to $\myvec{s}_j$ is
\begin{align}
    \frac{\partial \mathcal{H}}{\partial \myvec{s}_j}
    =
    \sum_{i=1}^{M}
    \int_{\mathcal{V}_i}
    f(r_i)\phi_j
    \boldsymbol{\Sigma}^{-1}(\myvec{q}-\myvec{s}_j)
    d\Omega .
\end{align}
Let $\hat{\phi}_{ij}=f(r_i)\phi_j$, and define $\hat{m}_{ij} = \int_{\mathcal{V}_i}\hat{\phi}_{ij}d\Omega$, $    \hat{\myvec{c}}_{ij} = \int_{\mathcal{V}_i}\myvec{q}\hat{\phi}_{ij}d\Omega/\int_{\mathcal{V}_i}\hat{\phi}_{ij}d\Omega$.
Then
\begin{align}
    \frac{\partial \mathcal{H}}{\partial \myvec{s}_j}
    =
    \boldsymbol{\Sigma}^{-1}
    \sum_{i=1}^{M}
    \hat{m}_{ij}(\hat{\myvec{c}}_{ij}-\myvec{s}_j).
\end{align}
Equivalently, by defining the weighted centroid
\begin{align}
    \myvec{c}^s_j
    =
    \frac{
    \sum_{i=1}^{M}\hat{m}_{ij}\hat{\myvec{c}}_{ij}
    }{
    \sum_{i=1}^{M}\hat{m}_{ij}
    },
\end{align}
the gradient can be written as
\begin{align}
    \frac{\partial \mathcal{H}}{\partial \myvec{s}_j}
    =
    \boldsymbol{\Sigma}^{-1}
    \sum_{i=1}^{M}
    \hat{m}_{ij}(\myvec{c}^s_j-\myvec{s}_j),
\end{align}
which gives \eqref{eq:gradHs3}.

\section{Jacobian of the Coupled Dynamics}
\label{appendixB}

Substituting the gradient descent-ascent controllers into the integrator dynamics yields
\begin{align}
    \myvec{f}(\myvec{x})
    =
    \begin{bmatrix}
    \dot{\myvec{P}}\\
    \dot{\myvec{S}}
    \end{bmatrix}
    =
    \begin{bmatrix}
    k^p(\myvec{C}^{p}-\myvec{P})\\
    k^s(\myvec{C}^{s}-\myvec{S})
    \end{bmatrix},
\end{align}
where $\myvec{x}=[\myvec{P},\myvec{S}]^\top$. The Jacobian has the block form
\begin{align}
    \nabla \myvec{f}
    =
    \begin{bmatrix}
    \frac{\partial \dot{\myvec{P}}}{\partial \myvec{P}} & \frac{\partial \dot{\myvec{P}}}{\partial \myvec{S}} \\
    \frac{\partial \dot{\myvec{S}}}{\partial \myvec{P}} & \frac{\partial \dot{\myvec{S}}}{\partial \myvec{S}}
    \end{bmatrix}
    =
    \begin{bmatrix}
    k^p\left(\frac{\partial \myvec{C}^p}{\partial \myvec{P}}-\myvec{I}_{2M}\right)
    &
    k^p\frac{\partial \myvec{C}^p}{\partial \myvec{S}}
    \\
    k^s\frac{\partial \myvec{C}^s}{\partial \myvec{P}}
    &
    k^s\left(\frac{\partial \myvec{C}^s}{\partial \myvec{S}}-\myvec{I}_{2N}\right)
    \end{bmatrix}.
\end{align}
Equivalently, the individual blocks are
\begin{subequations}
\begin{align}
    \left[\frac{\partial \dot{\myvec{P}}}{\partial \myvec{P}}\right]_{ik}
    &=
    k^p\left(
    \frac{\partial \myvec{c}^p_i}{\partial \myvec{p}_k}
    -
    \delta_{ik}\myvec{I}_2
    \right),\\
    \left[\frac{\partial \dot{\myvec{P}}}{\partial \myvec{S}}\right]_{ik}
    &=
    k^p
    \frac{\partial \myvec{c}^p_i}{\partial \myvec{s}_k},\\
    \left[\frac{\partial \dot{\myvec{S}}}{\partial \myvec{P}}\right]_{jk}
    &=
    k^s
    \frac{\partial \myvec{c}^s_j}{\partial \myvec{p}_k},\\
    \left[\frac{\partial \dot{\myvec{S}}}{\partial \myvec{S}}\right]_{jk}
    &=
    k^s\left(
    \frac{\partial \myvec{c}^s_j}{\partial \myvec{s}_k}
    -
    \delta_{jk}\myvec{I}_2
    \right).
\end{align}
\end{subequations}
Using the standard Leibniz rule for Voronoi partitions, the derivatives of the Voronoi centroids with respect to the agents' states are
\begin{subequations}
\begin{align}
    \frac{\partial\myvec{c}^p_i}{\partial\myvec{p}_i}
    &=
    -\sum_{k\in N_{\mathcal{V}_i}}\int_{\partial\mathcal{V}_{ik}}
    \frac{(\myvec{q}-\myvec{c}^p_i)(\myvec{p}_i-\myvec{q})^\top}{m_i\|\myvec{p}_k-\myvec{p}_i\|}\bar{\phi} \, d\Omega,\\
    \frac{\partial\myvec{c}^p_i}{\partial\myvec{p}_k}
    &=
    \int_{\partial\mathcal{V}_{ik}}\frac{(\myvec{q}-\myvec{c}^p_i)(\myvec{p}_k-\myvec{q})^\top}{m_i\|\myvec{p}_k-\myvec{p}_i\|}\bar{\phi} \, d\Omega,\qquad k\neq i .
\end{align}
\end{subequations}
The derivatives involving the adversary-dependent centroids are
\begin{subequations}
\begin{align}
    \frac{\partial \myvec{c}^s_j}{\partial \myvec{p}_k}
    &=
    -\frac{2}{\sum_{i=1}^{M}\hat{m}_{ij}}
    \int_{\mathcal{V}_k}
    \phi_j
    (\myvec{q}-\myvec{c}^s_j)
    (\myvec{q}-\myvec{p}_k)^\top
    d\Omega,\\
    \frac{\partial \myvec{c}^p_i}{\partial \myvec{s}_k}
    &=
    \frac{1}{m_i}
    \left(
    \int_{\mathcal{V}_i}
    \phi_k
    (\myvec{q}-\myvec{c}^p_i)
    (\myvec{q}-\myvec{s}_k)^\top
    d\Omega
    \right)
    \boldsymbol{\Sigma}^{-1},\\
    \frac{\partial\myvec{c}^s_j}{\partial\myvec{s}_j}
    &=
    \frac{1}{\sum_{i=1}^{M}\hat{m}_{ij}}
    \left(
    \sum_{i=1}^{M}
    \int_{\mathcal{V}_i}
    \hat{\phi}_{ij}
    (\myvec{q}-\myvec{c}^s_j)
    (\myvec{q}-\myvec{s}_j)^\top
    d\Omega
    \right)
    \boldsymbol{\Sigma}^{-1},\\
    \frac{\partial\myvec{c}^s_j}{\partial\myvec{s}_k}
    &=0,\qquad k\neq j .
\end{align}
\end{subequations}

\section{Hessian of the Game Payoff}
\label{appendixC}

We provide the expression of the Hessian matrix~\eqref{eq:hessian} of the game payoff.

The agent-agent block is given by
\begin{align}
    [\mathcal{H}_{pp}]_{ij}
    =
    \frac{\partial^2\mathcal{H}}
    {\partial\myvec{p}_i\partial\myvec{p}_j^\top}
    =
    -2m_i
    \left(
    \frac{\partial\myvec{c}^p_i}{\partial\myvec{p}_j}
    -
    \delta_{ij}\myvec{I}_2
    \right).
\end{align}
The adversary-adversary block is
\begin{subequations}
\begin{align}
    [\mathcal{H}_{ss}]_{jj}
    &=
    \sum_{i=1}^{M}
    \int_{\mathcal{V}_i}
    \hat{\phi}_{ij}
    (\myvec{q}-\myvec{c}^s_j)
    (\myvec{q}-\myvec{s}_j)^\top
    d\Omega
    -
    \myvec{I}_2,\\
    [\mathcal{H}_{ss}]_{jk}
    &=0,\qquad k\neq j .
\end{align}
\end{subequations}
The cross blocks are
\begin{subequations}
\begin{align}
    [\mathcal{H}_{ps}]_{ik}
    &=
    -2
    \left(
    \int_{\mathcal{V}_i}
    \phi_k
    (\myvec{q}-\myvec{c}^p_i)
    (\myvec{q}-\myvec{s}_k)^\top
    d\Omega
    \right)
    \boldsymbol{\Sigma}^{-1},\\
    [\mathcal{H}_{sp}]_{jk}
    &=
    -2
    \boldsymbol{\Sigma}^{-1}
    \left(
    \int_{\mathcal{V}_k}
    \phi_j
    (\myvec{q}-\myvec{c}^s_j)
    (\myvec{q}-\myvec{p}_k)^\top
    d\Omega
    \right).
\end{align}
\end{subequations}
At the fixed configuration in \eqref{eq:equilibrium}, these cross blocks satisfy
\begin{align}
    [\mathcal{H}_{ps}]_{ik}
    =
    [\mathcal{H}_{sp}]_{ki}^{\top}.
\end{align}

\bibliographystyle{plain}        
\bibliography{AreaCoverageGame}           

\end{document}